\documentclass[12pt]{article}
\usepackage{amsmath}
\usepackage{graphicx}
\usepackage{enumerate}
\usepackage{natbib}
\setcitestyle{authoryear,aysep={}}
\usepackage{url} 
\usepackage{amssymb}
\usepackage{amsthm}
\newtheorem{theorem}{Theorem}

\usepackage[utf8]{inputenc}
\usepackage{amsmath, amsfonts, amssymb, amsthm}
\usepackage{enumitem}
\usepackage{verbatim}
\usepackage{graphicx}
\usepackage{caption}
\usepackage{array}
\usepackage{bbm}
\usepackage{fancyvrb}
\usepackage{xcolor}
\usepackage{natbib}
\usepackage{url}
\usepackage{comment}

\usepackage{csquotes}
\usepackage{subcaption}
\usepackage{makecell}
\usepackage{placeins}
\usepackage[colorlinks=true, allcolors=blue]{hyperref}

\usepackage{tikz}
\usetikzlibrary{shapes.geometric, arrows}

\definecolor{darkgray}{rgb}{0.3, 0.3, 0.3}
\usepackage{hyperref}
\hypersetup{
   colorlinks=true,
    linkcolor=black,
    filecolor=magenta,      
    urlcolor=darkgray,
    pdftitle={Christopher Adolph :: Curriculum Vitae},
    bookmarks=false
}

\newtheorem{lemma}[theorem]{Lemma}
\newtheorem{defn}[theorem]{Definition}

\DeclareMathOperator*{\argmin}{arg\,min}

\newcommand{\tht}{\theta}

\usepackage{caption}
\usepackage{subcaption}

\newcommand{\rar}{\rightarrow}

\newcommand{\V}{\mathcal{V}}

\newcommand{\SSS}{\mathcal{S}}
\newcommand{\LL}{\mathcal{L}}

\newcolumntype{P}[1]{>{\centering\arraybackslash}p{#1}}

\begin{document}

\def\spacingset#1{\renewcommand{\baselinestretch}%
{#1}\small\normalsize} \spacingset{1.5}


\title{\large \textbf{Can statistical models capture Mamdani's success? Social choice, ranked-choice voting, and model fit, with an application to the 2025 New York City Democratic Primary}}

\author{
\small Michael Pearce$^\text{a}$, Erin R. Lipman$^\text{b}$, Christopher Adolph$^\text{cd}$, and Elena A. Erosheva$^\text{bde}$\\
\small $^\text{a}$Department of Mathematics and Statistics, Reed College\\
\small $^\text{b}$Department of Statistics, University of Washington\\
\small $^\text{c}$Department of Political Science, University of Washington\\
\small $^\text{d}$Center for Statistics and the Social Sciences, University of Washington\\
\small $^\text{e}$School of Social Work, University of Washington
}

\maketitle

\bigskip
\begin{abstract}
Ranked-choice voting is increasingly prevalent in elections. An extensive literature in \textit{social choice theory} -- the study of collective decision-making with the goal of making compromises among disparate opinions -- considers theoretical and empirical properties of various election procedures.
Recently, a sub-literature on \textit{rationalizability}, led by social choice theorists and computer scientists, makes explicit connections between social choice rules and maximum likelihood estimation. This paper further illuminates connections between social choice and statistical summaries of data. We begin by studying rationalizability from a statistical perspective, expanding existing results to realistic voting contexts and deriving statistical details necessary for model estimation and assessment. We then apply our work to ranked-choice votes from the 2025 New York City Democratic mayoral primary election. We demonstrate how rationalizing models impose unrealistic distributional assumptions and thus exhibit poor fit to voting data. Additionally, we show that non-rationalizing models meaningfully elucidate heterogeneous voter preferences. 
Our work demonstrates the inability of rationalizing models to capture key features of distributions of preferences in political elections, depriving analysts of fundamental uses, such as inference, typically associated with statistical modeling. Conversely, we show that statistical modeling can capture nuanced voter preferences by paying careful attention to plausible data-generating mechanisms. 
\end{abstract}

\noindent%

\spacingset{1.25}

\section{Introduction}\label{sec:introduction}

On the surface, social choice theory and statistical preference modeling are distinct in methods and aims. Social choice theory studies collective decision-making, providing formal rules to make compromises among disparate individual opinions. Statistical preference modeling is alternatively concerned with positing 
models for observed preferences, estimating their parameters and associated uncertainty, and thus inferring or predicting specific quantities of interest informed by the model. In the context of elections, social choice rules are chiefly concerned with yielding winner(s) whose selection can be justified with respect to desirable axioms. Statistical models instead aim to capture features of the distribution of individual opinions to allow for inference such as quantifying
 relative support for the candidates or identifying heterogeneous
voter blocs. Finally, social choice rules are deterministic, while statistical models aim to separate signal from noise under some posited distributional assumptions.

Despite these distinctions, social choice theory has relied on statistical thinking since its  origins in the 18th century. Important to understanding these connections is the distinction between settings where a ``ground truth” exists (\textit{epistemic} social choice) and those where it does not (\textit{non-epistemic}). For example, French thinker Nicholas de Condorcet posed his famous ``jury theorem" in the former setting, assuming individuals differed not in subjective preferences but in their assessment of a proposition's truth value~\citep{condorcet1785essay}. American economist Peyton Young similarly assumed a ``correct" choice exists when arguing for what is now called the Kemeny-Young social choice rule~\citep{young1995optimal}. In contrast, voters in elections often vary systemically based on their ideology or political party, creating the potential for paradoxes and perverse outcomes in the process of preference aggregation (e.g., Arrow's~(\citeyear{arrow1950difficulty}) impossibility theorem and McKelvey's~(\citeyear{mckelvey1976chaos}) chaos theorem).
Distinguishing between epistemic and non-epistemic settings is ultimately an exercise in ascertaining a plausible distribution of individual opinions. Assuming the existence of ``ground truth" admits the possibility that the distribution of individual opinions is plausibly unimodal, centered on that truth. That is, individual expressions may exhibit random noise around a true mode. Condorcet's jury theorem states that in such settings, one is likely to recover the ``truth" as the number of voters increases. On the other hand, when individual preferences reflect systematically heterogeneous ``tastes,'' plausible distributions of individual opinions are multi-modal. In this setting, opinions cannot be easily reconciled; a social choice cannot be ``correct" or ``incorrect'' because a single ground truth simply does not exist. 

Moreover, explicit connections between certain social choice rules and statistical distributions have been made in the past. Perhaps most famously, the ranking outputted by the Kemeny-Young rule is equivalent to the maximum likelihood estimator (MLE) of the modal ranking parameter in a Mallows distribution \citep{ali2012experiments}. 
In such cases, the social choice rule is said to be \textit{rationalized} by the statistical model~\citep{conitzer2012common}. One may be tempted to assume that a statistical model that rationalizes a social choice rule via its MLE may be used to conduct valid inference on votes, and thus improve our understanding of voter opinions or the social choice rule itself. 

Instead, we argue that for a statistical model rationalizing a social choice rule to be useful for inference, it first must fit the data well, just like any other model in any other application. Essential statistical considerations such as goodness-of-fit have too often been ignored when studying connections between social choice rules and statistical models. 
Obtaining a good absolute model fit to high-dimensional ranking data may be a tall order. Nevertheless, assessing distributional assumptions imposed by a rationalizing model and its fit to observed voting data with respect to meaningful discrepancy measures is necessary for valid statistical inference. 

In this paper, we study connections between social choice rules and rationalizing models from a statistical perspective. We demonstrate how rationalizing models often impose unrealistic distributional assumptions and exhibit poor fit. We also show that other, non-rationalizing statistical models may meaningfully elucidate voter preferences outside the social choice context of determining a winner.
The paper proceeds as follows: We first outline the history of social choice in epistemic and non-epistemic contexts and note existing connections to rationalizing statistical models (Section 2). We then review common social choice rules and statistical models for single-preference and ranked-choice voting contexts (Section 3). Next, we extend existing theoretical results on rationalizability to realistic election settings. This includes deriving a statistical distribution that rationalizes an important social choice rule for studying ranked choice elections, developing its properties, and proposing an algorithm for parameter estimation (Section 4). We use this machinery to analyze real ranked choice votes from the 2025 New York City Democratic mayoral primary election won by Zohran Mamdani. We demonstrate that constraints imposed by the ranked choice rationalizing model are stringent, unrealistic, and result in poor fit, both overall and specifically for first-place and second-place votes. We also study these data from a purely statistical perspective, demonstrating how a mixture model results in a well-fitting and interpretable representation of observed votes (Section 5). We conclude with a discussion of the importance of statistical modeling for studying voter preferences (Section 6).

\section{Social Choice Theory and Statistics}

In this section, we outline the history of social choice theory as it relates to the usual non-epistemic settings (in which there is no objectively ``correct" choice) and narrower epistemic contexts (in which we assume there is). We then discuss social choice rules that may be ``rationalized" by a statistical model and conclude by discussing practical problems of such rules.

\subsection{Epistemic Social Choice in Context}

Social choice theory is a broad field with a long history dating to 18th century French thinker Nicolas de Condorcet, best known for two eponymous contributions. \emph{Condorcet's paradox} considers situations where voters have three or more options and varied preferences thereof. He showed that even if voters have rational (transitive) preferences over the options, a sequence of pairwise majority votes can produce a cycle that neither identifies a single option preferred by the majority to all other alternatives nor produces a transitive social ordering of the options. 
\emph{Condorcet's jury theorem} holds when $n$ voters are posed a single, purely empirical question with two possible answers, correct and incorrect. Condorcet's jury theorem states that if each voter has an equal and independent chance $p$ of answering correctly, then the probability of the majority of voters reaching the correct answer approaches 1 as $n \to \infty$ for $p>\frac12$. The assumption that $p>\frac12$
is known in the literature as \textit{competence}. 

Condorcet's work illustrates that under strict conditions---sincere and competent voters united by homogeneous preferences---majority rule is suitable for aggregating votes into an optimal social choice. But if preferences are heterogeneous, then under a broad range of conditions, aggregation of individual preferences is beset with paradoxes and perverse (i.e., logically incoherent) outcomes. (The main exception is when every individual's preference ordering is single-peaked in $\mathbb{R}^1$, in which case the median voter theorem holds \citep{black1948}.)  \textit{Epistemic} social choice may be defined as theory which chooses to operate within the jury theorem's narrow assumptions \citep{list2022}. Applications of epistemic social choice theory share the assumption of a single ``ground truth" which individuals seek to determine. Thus, the social choice problem is aggregating opinions about the truth, rather than differences in personal preferences.\footnote{Even epistemic social choice is subject to impossibility results analogous to Condorcet's paradox and Arrow's impossibility theorem when there are multiple propositional claims to be decided. In such cases, majority rule can produce a set of logically incoherent outcomes \citep{listpettit2002}.} 

Most social decisions fall outside the scope of epistemic social choice. We note two ways this might occur. First, voters might be systematically ``incompetent'' on a particular binary question so that $p>\frac12$ does not hold.
Second, and more important for our paper, there may be no ``ground truth" for the aggregation rule to uncover. Human voters usually differ systematically in their preferences, rendering the notion of an objectively correct decision moot. This is hardly a controversial claim: the existence of varied preferences across actors has long been a core axiom in theories of economic and political choice \citep{arrow1951}. Experimentally validated cross-national surveys find substantial differences in individual preferences within and across countries \citep{falk2018} predicting, for example, varied occupational choices \citep{bonin2007} and health behavior \citep{chabris2008}. Finally, a long tradition in comparative political science characterizes voters' preferences over political parties as divided by enduring cleavages -- social class and religion \citep{lipsetrokkan1967}, ethnic identity \citep{horrowitz1985}, and vulnerability to globalization \citep{kriesi2008}, among others -- across which the concept of a ``ground truth" preference seems ill-defined.

The social choice literature in political science and economics has primarily focused on non-epistemic settings and has systematically found all aggregation mechanisms to exhibit significant limitations. The most influential social choice theorist since Condorcet is undoubtedly the American economist Kenneth Arrow, whose impossibility theorem~\citep{arrow1950difficulty} generalized Condorcet's Paradox. Arrow showed that no aggregation rule other than single-person dictatorship can be guaranteed to obtain a complete and transitive social ordering of alternatives for any possible set of transitive voter preferences while still satisfying both the weak Pareto property (if all voters prefer $x$ to $y$, then the social ordering ranks $x$ above $y$) and the independence of irrelevant alternatives (the relative social ordering of $x$ and $y$ depends only on voters' rankings of these options). Relatedly, the chaos theorem~\citep{mckelvey1976chaos} shows that when voters' ideal outcomes map to a latent space on $\mathbb{R}^k, k>1$, then, except under irrelevant measure zero cases \citep{plott1967}, pairwise majority rule does not ensure a Condorcet winner but instead can lead to any outcome whatsoever, including points outside the convex hull of voters' preferences. As a result of these and other theoretical results, for the past half century, social choice theorists have mostly abandoned the search for an ``ideal" aggregation mechanism for social choice in democratic settings. Instead, some have endeavored to understand how political institutions
enable strategic actors to achieve desired outcomes through voting procedures \citep{shepsleweingast1981sie}, or seek ``possibilist" compromises \citep{sen1998} by investigating the theoretical tradeoffs created by aggregation rules that relax various Arrowian axioms, as well as the empirical significance of these tradeoffs, e.g., \cite{list2013deliberation,miller1992deliberative,knight1994aggregation,dryzek2003social,List2001epistemic,tsetlin2003impartial,regenwetter2006behavioral}.

\subsection{Social Choice Rules as MLEs}

Recent work in social choice theory often seeks to characterize unique solutions to social choice problems based on specific sets of axioms.  Evaluating such work entails not only proving the existence or uniqueness of social choice solutions under specific axioms, but also defending the choice of those axioms under some scope. An important example of this axiomatic approach to social choice emphasizes \textit{rationalizability}. A social choice rule is rationalizable if it admits a statistical model whose maximum likelihood estimator (MLE) of a group preference parameter\footnote{We only consider rationalizability with respect to an MLE. Other forms of rationalizability exist (e.g., maximum \textit{a posteriori} rationalizability) which are less common and beyond the scope of this paper.} is identical to the social choice whenever the rule and the model are supplied the same votes \citep{pivato2013voting, conitzer2012common, xia2014statistical}. One might justify a social choice rule by proving that it is rationalized by some statistical distribution of individual opinions.

The literature on rationalizability largely focuses on epistemic social choice problems, such as aggregating the outcomes of autonomous software agents solving a problem (e.g., internet search) which arguably has a ground truth \citep{brandtetal2016}. However, some papers arguing for the selection of social choice rules on the grounds of rationalizability conjecture applications to explicitly political decisions, such as how central banks set interest rates \citep{young1995optimal}, by claiming they involve disagreement only on factual questions, rather than preferences over outcomes. However, there are strong reasons to doubt such questions are ever merely epistemic, rather than matters of widely differing preferences over policy objectives \citep{adolph2013banks,jacobsking2021fedpower}.

Economist Peyton Young makes the most expansive case for rationalizability as a justification for social choice rules outside epistemic problems. In particular, he claims that simple majority rule is ``the best way to estimate the optimal choice'' on two alternatives and that this result ``applies to any choice problem in which people agree about the objective, but disagree about the best means to achieve that objective''~\citep[p. 53]{young1995optimal}.
Furthermore, he argues that regardless of whether a ground truth exists, social choice rules which produce group decisions that are most likely to be ``correct" under that assumption should be preferred in practice.\footnote{Young argues that for Condorcet, like Jean-Jacques Rousseau, voting is an exercise in finding the ``truly best" option. Though Young acknowledges that what makes a choice ``best" is undefined, he argues that it is nonetheless a realistic criterion in cases like a jury's determination of guilt or innocence, or an expert panel's prediction of an ``observable event" \citep[p. 1232]{young1988}.}  
For cases in which more than two alternatives are ranked, Young assumes that between each pair of alternatives, each individual makes the ``correct'' choice more than half the time (where the existence of a ``correct'' choice assumes the existence of a factual answer or a common preference); the individual's choice between any pair of alternatives is independent of their choices between the other pairs; and that choices of different individuals are independent.
These distributional assumptions imply that individual rankings follow a Mallows distribution \citep{mallows1957non,ali2012experiments}. Accordingly, \citet{young1995optimal} concludes that the Kemeny-Young ranking---which is the MLE of a Mallows distribution---is the most appropriate solution to the social choice problem on statistical grounds.

Pointing out other desirable properties of the Kemeny-Young ranking---that it is anonymous, neutral, Pareto efficient, and satisfies reinforcement and local independence of irrelevant alternatives---Young (\citeyear{young1995optimal}, p. 62-63) concludes that the maximum likelihood ranking ``is arguably the best method if we think of voting as a collective quest for truth$\ldots$ on the other hand, in situations where voting appears to be a way of compromising between conflicting values, maximum likelihood rule still makes sense because it represents the median opinion."
 
Although, as Young (\citeyear{young1995optimal}, p. 52) himself acknowledges, these conclusions differ significantly from most discussions of social choice under conflicting preferences, his proposal raises the question of whether MLE-based social choice could work well in practice.

\subsection{Practical problems with MLE-based social choice}

Fast and accurate computation is an important consideration in social choice and statistical estimation. Social choice rules commonly used in practice, such as plurality rule and instant runoff voting, can be obtained by following simple algorithms. However, for the Kemeny-Young ranking and corresponding MLE of a Mallows distribution, no such algorithm exists. In fact, there may not even be a unique solution. Moreover, when ranking data are generated from a distribution that does not have a mode (e.g., a discrete uniform across rankings) or is a mixture of groups with different opinions, the solution is highly unstable such that adding or removing a single observation can have a large influence \citep{ali2012experiments}.

Importantly, many reasonable social choice rules are not rationalizable. For example, under certain conditions the ranked-choice (single winner), single transferable vote (multi-winner), Buckland, Copeland, maximin, and ranked pair rules do not admit a corresponding statistical model \citep{conitzer2012common}.

Even when rationalizing models exist, the literature on statistical modeling recognizes that a single unimodal distribution may not be able to capture structural aspects of variability among voter expressions. For example, early work on computing the Kemeny ranking noted that, in practice, mixture distributions need to be considered for real ranking data~\citep{meila2007}. Analyzing actual ballots of single transferable votes from Irish elections, mixtures of Benter distributions were used for capturing Irish electorate's voting blocs \citep{Benter1994,gormley2008exploring}. However, while useful for describing heterogeneity, mixture models do not lend themselves to a natural mapping onto a social choice rule yielding a single rank-ordering of candidates.

\section{Social Choice and Statistical Perspectives on Preference Aggregation}\label{sec:background}

In this section, we review social choice rules and statistical models that may be used to analyze single preference and ranked choice votes. Then, we draw specific connections among them via rationalizability.

\subsection{Social Choice Rules}\label{sec:socialchoice}

Formally, social choice rules are deterministic functions that input a set of individual votes and output the corresponding group preference. 
Individual votes may take various forms, such as a \textit{single preference} for one candidate, a \textit{partial ranking} of a candidates, or a \textit{complete ranking} of all candidates.
The group preference may also take various forms. Most commonly a single winner is determined, although some elections result in multiple winners. A complete ranking of candidates from best to worst is also a valid group preference, even if it is not often an actual election outcome.

Let $X$ be a finite set of candidates and $\mathcal{L}(X)$ be the set of partial orderings on $X$. Let $\V\subseteq \LL(X)$ be the allowable set of votes. 
A set of $N$ votes, $V=\{v_1,\dots,v_N\}\in \V^N$ is called a \textit{profile}. Let $\SSS\subseteq \LL(X)$ be the set of allowable outcomes. A social choice function is a map $f: \V^N\rar\SSS$ from a profile of $N$ votes to an outcome.

We now present 5 social choice rules commonly used in elections.
The first two rules elect a single winner; the remaining three determine outcomes beyond a single winner.

\vspace{1em}
\noindent\textbf{Example 1 (Plurality)} \textit{Suppose voters cast single preference votes among $J$ candidates. The winner is the candidate who received more votes than any other.}
\vspace{1em}

\noindent\textbf{Example 2 (Instant Runoff Voting [IRV])} \textit{Suppose voters cast partial or complete rankings among $J$ candidates. The winner is chosen in multiple rounds as follows (illustrated with example votes in Figure~\ref{fig:IRV_ex}):}

\vspace{.5em}
\textit{1. If any candidate receives a majority of first place votes, he/she wins. 
    In the example, candidates $a$, $b$, and $c$ receive $4$, $6$, and $3$ first-choice respectively. Thus no candidate receives the $7$ out of $13$ votes required for a majority (Figure~\ref{fig:IRVa}).}
    
\textit{2. If no candidate receives a majority, the candidate with the fewest first-choice votes is eliminated. This candidate is removed from each voter's ranking; all candidates ranked behind the eliminated candidate move up one place in the ranking. }

\textit{3. Steps 1 and 2 are repeated until some candidate holds the majority of the first-choice votes. That candidate is declared the winner.
    In the example, $a$ holds the majority of first-choice votes (7 of 13) after $c$ is eliminated (Figure~\ref{fig:IRVb}), and thus wins.}

\vspace{1em}
\noindent\textbf{Example 3 (Single Transferable Voting for ranking all candidates [STV-R])} \textit{Voters express a ranking of some or all candidates. A ranking of all candidates is determined over multiple rounds as follows (illustrated again by votes in Figure~\ref{fig:IRV_ex}):}

\vspace{.5em}
\textit{1. The candidate with the fewest first place votes is eliminated. This candidate is removed from each voter's ballot; all candidates ranked behind the eliminated candidate move up one place in the ranking. In the example, candidate $c$ is eliminated first.}
    
\textit{2. Step 1 is repeated until no candidates remain. The reverse elimination order determines the ranking of candidates. In the example, candidate $b$ is eliminated second and candidate $a$ is eliminated last. Hence, the social ranking is $a\succ b\succ c$.}

\begin{figure}[h!]
    \centering
    \begin{subfigure}[t]{0.25\textwidth}
    \centering
    \begin{tabular}{c|c}
         \# votes & $v$ \\
         \hline
         3 & $c\succ a\succ b$\\
         4 & $a\succ b\succ c$\\
         6 & $b\succ a\succ c$\\
    \end{tabular}
    \caption{Full profile}
    \label{fig:IRVa}
    \end{subfigure}
    ~$\overset{\text{remove c}}\Longrightarrow$~
    \begin{subfigure}[t]{0.25\textwidth}
    \centering
    \begin{tabular}{c|c}
         \# votes & $v$ \\
         \hline
         3 & $a\succ b$\\
         4 & $a\succ b$\\
         6 & $b\succ a$\\
    \end{tabular}
    \caption{Profile after $c$ is eliminated}
    \label{fig:IRVb}
    \end{subfigure}
    ~$\overset{\text{remove b}}\Longrightarrow$~
    \begin{subfigure}[t]{0.25\textwidth}
    \centering
    \begin{tabular}{c|c}
         \# votes & $v$ \\
         \hline
         3 & $a$\\
         4 & $a$\\
         6 & $a$\\
    \end{tabular}
    \caption{Profile after $b$ and $c$ are eliminated}
    \label{fig:IRVc}
    \end{subfigure}
    \caption{Illustrative example for IRV and STV-R}
    \label{fig:IRV_ex}
\end{figure}

\noindent\textbf{Example 4 (Single Transferable Voting for electing $K$ candidates [STV-K])} \textit{Voters express a partial or complete ranking of $J$ candidates. $K<J$ candidates are elected over multiple rounds as follows:}

\vspace{.5em}
\textit{1. A candidate is declared a winner if they receive more than $\frac{N}{K+1}$ first place votes, where $N$ is the number of voters. If one or more candidates exceed this threshold, excess votes are transferred to those voters' next-choice candidate(s).}
    
\textit{2. If fewer than $K$ candidates have been elected, the candidate with the fewest first place votes is eliminated. This candidate is removed from each voter's ballot; all candidates ranked behind the eliminated candidate move up one place in the ranking.}

\textit{3. Steps 1 and 2 repeat until $K$ candidates are elected.}

\vspace{1em}
\noindent\textbf{Example 5 (Kemeny-Young)} \textit{Voters express a partial or complete ranking of $J$ candidates. The social choice rule outputs a ranking of all candidates according to
\[\argmin_{s\in\mathcal{S}} \sum_{i=1}^N \tau(v_i,s),\]
where $\tau(v,s)$ is the Kendall-Tau distance between two rankings $v$ and $s$ over $X$, i.e., the number of pairs $x,y\in X$ for which $v$ and $s$ rank in different relative orderings.}

\subsection{Statistical Preference Models}\label{sec:statmodels}

Unimodal statistical models for preference data posit there exist parameters describing the population preference and assume each voter's preference is a noisy version of the population preference. Given a set of candidates $X$, a vote space $\V\subseteq \LL(X)$, and a domain of population preferences $\SSS\subseteq \LL(X)$, let $s\in\SSS$ represent the true population preference and $V=(v_1,\dots,v_N)$ the sampled profile of votes. Then voters' preferences are described by a statistical model $P(V\mid s)=P(v_1,\ldots,v_N\mid s)$. In what follows, we assume that all votes are independent and identically distributed given the true (but unknown) population preference $s$ and model parameters $\theta$, that is,
\[p(V\mid s, \tht)=\prod_{i=1}^N p(v_i\mid s,\tht).\]

Four commonly-used statistical models for preference modeling are the Multinomial, Mallows \citep{mallows1957non}, Plackett-Luce \citep{plackett1975analysis}, and Benter \citep{Benter1994}. We define each in turn.

\vspace{1em}
\noindent\textbf{Example 6 (Multinomial)} \textit{Let candidates $x_j \in X$, $j=1,\dots,J$, be assigned a worth parameter $\theta_j>0$ such that $\sum_j\theta_j=1$ and define $s\equiv\text{order}(\theta)$. The Multinomial model assumes that votes in the form of a top-1 ranking are generated IID according to the likelihood
\begin{equation}
    p(V|\theta) = \prod_{i=1}^N\prod_{j=1}^J \theta_j^{I\{v_i(1)=j\}} \label{eq:Multinomial}
\end{equation}
\noindent where $v_i(1)$ is the candidate in first place according to voter $i$. }
\vspace{1em}

\vspace{1em}
\noindent\textbf{Example 7 (Mallows)} \textit{The Mallows model assumes that votes are generated IID according to the likelihood
\begin{equation}
    p(V\mid s, \theta)\propto
\prod_{i=1}^N \exp\left(-\theta\cdot \tau(v_i,s)\right)
=\exp\left(-\theta\sum_{i=1}^N \tau(v_i,s)\right) \label{eq:Mallows}
\end{equation}
\noindent where $\tau$ denotes the Kendall-Tau distance. A dispersion parameter $\theta>0$ controls the concentration of rankings around population preference $s$.}
\vspace{1em}

\vspace{1em}
\noindent\textbf{Example 8 (Plackett-Luce)} \textit{Let candidates $x_j \in X$, $j=1,\dots,J$, be assigned a worth parameter $\theta_j>0$ such that $\sum_j\theta_j=1$ and define $s\equiv \text{order}(\theta)$. The Plackett-Luce model assumes that votes are generated IID according to the likelihood
\begin{equation}
    p(V|\theta) = \prod_{i=1}^N\prod_{r=1}^J \frac{\theta_{v_i(r)}}{\sum_{s=r}^J \theta_{v_i(s)}}\label{eq:PL}
\end{equation}
\noindent where $v_i(r)$ is the candidate in $r^\text{th}$ place according to voter $i$.}
\vspace{1em}

\vspace{1em}
\noindent\textbf{Example 9 (Benter)} \textit{Let candidates $x_j \in X$, $j=1,\dots,J$, be assigned a worth parameter $\theta_j>0$ such that $\sum_j\theta_j=1$  and define $s\equiv \text{order}(\theta)$. Let $\alpha_r\in[0,1]$, $r\in[1,J]$ be a dampening parameter for rank-level $r\in[1,J]$. The Benter model assumes that votes are generated IID according to the likelihood
\begin{equation}
    p(V|\theta,\alpha) = \prod_{i=1}^N\prod_{r=1}^J \frac{\theta_{v_i(r)}^{\alpha_r}}{\sum_{s=r}^J \theta_{v_i(s)}^{\alpha_r}}\label{eq:Benter}
\end{equation}
\noindent where $v_i(r)$ is the candidate in $r^\text{th}$ place according to voter $i$. Usually, $\alpha_1=1$ for identifiability. }
\vspace{1em}

Statistical models are importantly judged by their ability to fit data. The Multinomial model is likely to adequately fit single-preference votes because of the simplicity of the vote itself. The remaining models pose a more challenging problem. The Mallows model assumes that all rankings equal in distance from $s$ have the same probability. 
This assumption may be highly unrealistic in ranked-choice elections where voters typically have clear preferences regarding the most prominent candidates (e.g., those typically ranked first or second) but are more likely to randomly swap the ordering of less serious and less known candidates (e.g., those commonly ranked near the bottom of voters' lists). The Plackett-Luce model relaxes the symmetry assumption but imposes, among other assumptions, ``independence of irrelevant alternatives" (IIA). The IIA assumption may be unrealistic in the case of voter heterogeneity. The Benter model further relaxes assumptions of Plackett-Luce by adding a dampening parameter that permits additional randomness at down-ballot rank levels.

The Mallows, Plackett-Luce, and Benter models contain relatively few parameters to capture a high-dimensional and likely multi-modal collection of votes. Thus, it may be necessary to fit these models within the context of a latent class mixture model framework to achieve acceptable goodness-of-fit. In a latent class mixture model, we assume votes are generated IID according to the likelihood,
\begin{equation}
    p(V|s,\theta,\pi) = \prod_{i=1}^N\prod_{k=1}^K \pi_k \times p(v_i|s_k,\theta_k), \label{eq:mixtures}
\end{equation}
where $K$ is the number of latent classes, $\pi_k$ is the probability that each voter belongs to class $k\in\{1,\dots,K\}$ such that $\sum_{k=1}^K\pi_k=1$, and $(s_k,\theta_k)$ are class-specific parameters for each class $k$. 

In practice, the number of latent classes, $K$, must be selected. Two common selection criteria are the Bayesian Information Criterion (BIC) and Integrated Complete-data Likelihood (ICL). We define BIC and ICL according to,
\begin{align}
\begin{split}
    \text{BIC} &= -2\hat{L} + \nu\log(n)\\
    \text{ICL} &= -2\hat{L} + \nu\log(n) -2\sum_{i=1}^n\sum_{k=1}^K\hat z_{ik}\log\hat\pi_{ik}
\end{split}\label{eq:bicicl}
\end{align}
where $\hat{L}$ is the maximized observed model loglikelihood, $\nu$ is the number of model parameters, $n$ is the number of votes, $\hat{z}_{ik}$ is a binary indicator variable for voter $i$ being assigned to class $k$, and $\hat\pi_{ik}$ is the estimated probability that voter $i$ belongs to class $k$. BIC is a common model selection criterion but is known to overestimate the number of latent classes when applied to mixture models. ICL aims to further encourage parsimony by adding an entropy term that penalizes models with overlapping latent classes. Models with low BIC and ICL are generally preferred \citep{celeux2019model}.

\subsection{Connecting Social Choice and Statistical Modeling via Rationalizability}\label{sec:connections}

Any statistical estimator $\hat{s}$ for population preference $s$ can be seen as a social choice rule \citep{pivato2013voting, conitzer2009preference}. However, as noted earlier many social choice rules are rationalized by a statistical model and corresponding maximum likelihood estimator.
The MLE of the Multinomial model, 
$$\hat{s} = \text{order}(\hat\theta), \text{ where } \hat\theta_j = N^{-1}\sum_{i=1}^N I\{v_i(1)=j\},$$
is the ranking of candidates under an election based on plurality rule.
The MLE of a Mallows model, $$\hat{s}=\underset{s}{\arg\min} \sum_{i=1}^N \tau(v_i,s),$$ regardless of $\theta$ is the ranking of candidates under an election based on Kemeny-Young.
Some social choice rules are rationalized by statistical distributions lacking names. For example, \cite{conitzer2012common} proposed probability distributions that rationalize the Borda and STV-R rules. 

Conversely, \cite{conitzer2012common} provide the following lemma for demonstrating that a rule is not rationalizable:
\begin{lemma}
\label{lem:reinforcment}
    \textit{For a voting rule $f$, if there exist two voter profiles $V_1$ and $V_2$ such that $f$ chooses the same outcome for both profiles, $f(V_1)=f(V_2)$, but produces a different outcome for the combined profile, then $f$ is not rationalizable.}
\end{lemma}
They use this lemma to prove that the IRV, STV-K, Buckland, Copeland, Maximin and Ranked Pair rules are not rationalizable under certain conditions. In the reverse direction, no known social choice rule is rationalized by a Plackett-Luce or Benter distribution.

\section{New Results on Ranked Choice Voting}\label{sec:RCVoting}

In this section, we extend results on non-rationalizability of common ranked choice election rules to realistic election scenarios. We subsequently extend an existing STV-R rationalizing model proposed by \cite{conitzer2012common} to partial vote elections. We further develop their model by stating a complete model likelihood, proposing an efficient estimation procedure, and providing tools for goodness-of-fit assessment.

\subsection{Extending non-rationalizability results to realistic election scenarios}
\cite{conitzer2012common} show that IRV is not MLE-rationalizable for an election with $J=3$ candidates. Due to its relevance for our case study in Section \ref{sec:CaseStudy}, we minimally extend their result to the case when $J\geq3$.
\begin{theorem}\label{thm:irv} Consider a single-winner election with $J\geq 3$ candidates, $a_1,\dots,a_J$, where $\V$ is the set of full rankings. Then, the IRV social choice rule is not rationalizable under an IID statistical model.
\end{theorem}

Next, note that STV-K generalizes IRV to multi-winner elections. However, no results on rationalizability of STV-K elections exist. The following theorem proves non-rationalizability of STV-K under certain conditions.
\begin{theorem}\label{thm:stv-k}
Consider an election with $J\geq 3$ candidates where $K$ candidates, such that $1\leq K\leq J-2$, are to be elected and $\V$ is the set of full rankings. Then, the STV-K social choice rule is not rationalizable under an IID statistical model.\footnote{When $K=1$, STV-K is equivalent to IRV, which was proven to be non-rationalizable. Conversely when $K=J$, there is no meaningful election as all candidates are elected. We leave open the question of rationalizability when $K=J-1$, noting this circumstance is unlikely to occur in practice.}
\end{theorem}

Proofs of theorems \ref{thm:irv} and \ref{thm:stv-k} straightforwardly apply lemma~\ref{lem:reinforcment} via counterexample and are thus relegated to the appendix.

\subsection{A complete statistical model rationalizing STV-R}\label{sec:STVR}
As noted earlier, STV-R is rarely used as an election procedure. But it is often used to study results of ranked choice elections. \cite{conitzer2012common} proved the following model rationalized STV-R:
\begin{equation}\label{model:contizer_stvr}
    \begin{aligned}
    p_\text{CS}(V|s,k) &\propto \prod_{i=1}^N\prod_{j=1}^J k_j^{\delta_j(v_i,s)}\\
    \delta_j(v_i,s) &= I\big\{v_i(1)=s_j \text{ after removal of }s_{j+1},\dots,s_{J} \text{ from } v_i\big\}\\
    1&>k_1 >> k_2 >> \dots >> k_J> 0,
\end{aligned}
\end{equation}
where each vote $v_i$ is a complete ranking of the $J$ candidates, $s$ is the social choice based on STV-R, and $k\in(0,1)^J$ are the relative candidate strengths.

As it will prove useful, we paraphrase an explanation by \cite{conitzer2012common} for why the above model rationalizes STV-R: For fixed and well-separated values of $k$, the likelihood $p_{CS}(V|s,k)$ is dominated by factors with $k_J$. Thus to maximize the likelihood, one must minimize the number of times $k_J$ appears. Mathematically, we want a ranking $s$ that minimizes $\sum_{i=1}^N\delta_J(v_i,s)$. This occurs when the candidates with the fewest first-place votes appears last in $s$, whom we call $s_J$. Note that in STV-R, $s_J$ is eliminated first. Subsequently, the remaining terms in the likelihood are dominated by the factors with $k_{J-1}$; we seek to minimize the number of times $k_{J-1}$ appears by finding an $s$ (given that candidate $s_J$ is last in $s$) that minimizes $\sum_{i=1}^N\delta_{J-1}(v_i,s)$. But this is precisely the candidate with the fewest first-place votes after transferring votes cast for $s_J$ to those voters' second-place candidates. Again, STV-R eliminates this same candidate second. The process continues similarly until a complete ranking $s$ of candidates is obtained that matches the STV-R ranking.

The above rationalizing model has important limitations. 
First, it applies only to votes as complete rankings, whereas partial rankings are common in elections with many voters or candidates.
Second, it is unclear how one can conduct parameter estimation as both the likelihood and separability constraints on $k$ are not fully specified.
Third, \cite{conitzer2012common} provide no means for assessing model goodness-of-fit. 
We fill in those gaps, starting with a precise statistical specification of a rationalizing model that applies to elections under partial rankings.

\begin{defn}[STV-R rationalizing model under partial votes]
    Let $X$ be a set of $J$ candidates in an election and $\V_r$ the set of partial rankings over $X$ of length $r$, $1\leq r\leq J$. Let $V$ be a profile of $N$ partial rankings, where vote $v_i$ is of length $r_i$. Then, define the probability of observing profile $V$ according to:
    \begin{equation}\label{model:generalized_stvr}
    \begin{aligned}
    p(V|s,k) &= \prod_{i=1}^N\frac{\prod_{j=1}^J k_j^{\delta_j(v_i,s)}}{C_{r_i}(k)}\\
    \delta_j(v_i,s) &= I\big\{v_i(1)=s_j \text{ after removal of }s_{j+1},\dots,s_{J} \text{ from } v_i\big\},\\
    C_r(k) &= \sum_{\delta\in\Delta}m_\delta(r)\prod_{j=1}^J k_j^{\delta_j},\\
    1 &> k_1 > k_2 > \dots > k_J > 0,
\end{aligned}
\end{equation}
where $s$ is the population candidate ranking, $k\in(0,1)^J$ is the vector of candidate strengths, $\Delta=\{0,1\}^J$ is the set of distinct possible values for $\delta$,  and $m_\delta$ is the number of votes $v'\in\V_r$ in equivalence class $\delta$.
\end{defn}

The above model applies to elections under partial votes, including those where the length of each vote varies. We assume the number and length of each vote is fixed and known. Note that $p(v|s,k)$ constitutes a probability distribution: $p(v|s,k)\geq0$ for all votes $v\in\V$ as the product of nonnegative numbers, and
$$\sum_{v\in\V_r} p(v|s,k) = \sum_{v\in \V_r}\frac{\prod_{j=1}^J k_j^{\delta_j(v_i,s)}}{C_{r_i}(k)}= \sum_{\delta\in \Delta}m_\delta(r)\frac{\prod_{j=1}^J k_j^{\delta_j}}{C_{r_i}(k)}=\frac{\sum_{\delta\in \Delta}m_\delta(r)\prod_{j=1}^J k_j^{\delta_j}}{\sum_{\delta\in\Delta}m_\delta(r)\prod_{j=1}^J k_j^{\delta_j}}=1.$$

For reasons which will become clear when considering goodness-of-fit, we next explain a peculiar property of the STV-R rationalizing model via the following lemma.
\begin{lemma}\label{lem:Criteria2} Assume the model in equation \ref{model:generalized_stvr} and fix $s$ and $k$. For votes $v_1,v_2\in\mathcal{V}$ of equal length $r$, if $v_1(1) = v_2(1) = s_1$ then $p(v_1|s,k) = p(v_2|s,k)$. That is, all votes which rank the overall winner in first place have equal probability.
\end{lemma}
\begin{proof}
    Let $v\in\V_r$ be a vote such that $v(1)=s_1$. Then, $\delta(v,s)=(1,0,0,0,\dots,0)$ regardless of $v(2),\dots,v(J)$ by definition of $\delta_j$. Consequently, the probability mass function $p(v|s,k)$ is equivalent for all such votes.
\end{proof}

We proceed with maximum likelihood estimation of $s$ and $k$. 
First, note that $C_r(k)$ is a normalizing constant that depends only on $J$, $r$, and $k$ and is independent of parameter $s$ and observed votes $V$. Thus $C_r(k)$ is efficient to compute since calculation of $m_\delta(r)$ needs only be performed once for each combination of $J$ and $r$.
Second, we set $k_1$ to $e^{-0.001}$ for identifiability. Any fixed value for $k_1$ in the unit interval leads to an identifiable model; the value $e^{-0.001}$ was chosen simply because it is close to 1.
Third, to ensure the STV-R ranking $s$ is the MLE $\hat{s}$, we require sufficient separation of values $k$. A sufficient condition is to require $(\prod_{\ell=1}^{j-1}k_\ell)^N > k_j$ for all $j=2,\dots,J$. However, this condition is stringent and can often be relaxed to permit better fit to data; see appendix for details.
Fourth, one may straightforwardly determine the MLE $\hat{k}$ given the previous considerations, where $\hat{s}=\text{order}(\hat{k})=s$.

We conclude this section with a discussion of goodness-of-fit. To assess model fit with discrete rankings, a typical approach is to compare the observed counts of each possible ranking to the expected counts under a fitted model \citep[p. 34]{agresti1996introduction}. However, when $J\geq 5$, this form of assessment is impractical due to the large domain of the data. Instead, one may group rankings into lower order marginals relevant to the model and similarly compare observed and expected counts \citep{cohen1983assessing}. This technique has been adapted by, e.g., \cite{feigin1978model,yu2000bayesian}, and \cite{mollica2022remarkable}. 

In that spirit, we propose assessing goodness-of-fit via the following three criteria:\footnote{Details of how to calculate the observed and expected counts for the STV-R rationalizing model are provided in the appendix.}
\begin{enumerate}
    \item \ For each candidate $s_j$, $j=1,\dots,J$, compare the observed number of first place votes $N_j^{(1)}$ and the expected number of first place votes $N_j^{*(1)}$ under the fitted model.
    \item \ Among votes where candidate $s_1$ is in first place, for each candidate $s_j$, $j=2,\dots,J$, compare the observed number of second place votes $N_j^{(2|s_1)}$ and the expected number of second place votes $N^{*(2|s_1)}$ under the fitted model.
    \item \ For each pair of candidates $s_i, \ s_j,\ i\neq j$, compare the observed number of votes in which $s_i$ is ranked above $s_j$ to the expected number under the fitted model.
\end{enumerate}
Criterion 1 assesses how a model captures first-place votes, which are important in STV-R elections because candidates are eliminated largely based on the number of first-place votes received. Similar criteria appear in \cite{yu2000bayesian} and \cite{mollica2017bayesian}.
Criterion 2 partially assess how well observed votes satisfy the assumption stated in Lemma \ref{lem:Criteria2}. Although criterion 2 was developed to test this specific assumption of the STV-R rationalizing model, it may be still be applied to other statistical models.
Criterion 3 assesses how well a model captures voters' pairwise preferences among candidates. There are $\frac{J(J-1)}2$ such pairs in an election with $J$ candidates. Thus, criterion 3 provides a more granular metric of model fit. Similar criteria appear in \cite{cohen1983assessing}, \cite{mollica2017bayesian}, and \cite{pearce2024bayesian}.

The concordance between observed and expected counts for each criteria may be formally assessed by conducting a hypothesis test, e.g., a Pearson chi-squared goodness-of-fit test. In the presence of partial rankings, we note that computing a null distribution for the p-value calculation is challenging. 
Further, since elections often comprise thousands (if not millions) of votes, we should expect p-values to be extremely small even when data reasonably satisfies a model's assumptions \citep{cohen1983assessing}. Instead, to identify severe deviations, we heuristically assess model fit on the basis of criteria 1--3 by visually comparing observed and expected counts.

\section{Case Study: 2025 New York City Democratic Mayoral Primary}\label{sec:CaseStudy}

The June 2025 New York City (NYC) Democratic mayoral primary election pitted an array of progressive candidates---including State Assemblymember Zohran Mamdani, Comptroller Brad Lander, City Council Speaker Adrienne Adams, former Comptroller Scott Stringer, state Senator Zellnor Myrie, and former state Assemblymember Michael Blake---against a more centrist option, former New York state Governor Andrew Cuomo. In the election, voters could rank up to 5 candidates from a pool of 11 (or write in an undeclared candidate). Then, a winner was tabulated using a minor variant of the standard IRV procedure in which undeclared write-in candidates were batch eliminated in the first round. Cuomo consistently led Mamdani in opinion polls until the month of the election,\footnote{Nick Reisman,``Despite missteps, Andrew Cuomo maintains commanding lead in New York mayor’s race," \emph{Politico}, 14 May 2025, \url{https://www.politico.com/news/2025/05/14/andrew-cuomo-lead-poll-mayor-00347151}.} when Mamdani and other progressive candidates turned the ranked choice voting procedure to their advantage. Progressive candidates formed alliances encouraging their supporters to rank multiple members of the progressive bloc, while simultaneously urging voters not to rank Cuomo at all.\footnote{For example, Mamdani and Lander each recommended their voters rank the other candidate second (Emily Ngo, ``Zohran Mamdani, Brad Lander are cross-endorsing in race for New York City mayor," \emph{Politico}, 13 June 2025, \url{https://www.politico.com/news/2025/06/13/mamdani-lander-cross-endorsing-mayor-00405357}). This alliance paid off for both Mamdani, who was elected mayor in November 2025, and Lander, who rode Mamdani's endorsement to victory over Democratic incumbent Dan Goldman in the June 2026 Democratic primary for New York's 10th Congressional district (Joe Anuta and Bill Mahoney, ``The powerbroker: Mamdani-endorsed candidates obliterate old guard Dems in NYC,'' \emph{Politico}, 24 June 2026, \url{https://www.politico.com/news/2026/06/24/the-powerbroker-mamdani-endorsed-candidates-obliterate-old-guard-dems-in-nyc-00974317}).} Mamdani won the election, beating Cuomo in the final round after all other candidates were eliminated.

In our analysis, we remove all undeclared write-in candidates, mimicking the first round of the actual election. In total, $1{,}071{,}081$ voters ranked at least one declared candidate. Among those voters, 22\% ranked only one candidate and 47\% ranked the maximum five. The remaining 31\% of voters were approximately evenly divided between ranking two, three, and four candidates.

\begin{table}[h]
    \centering
    \begin{tabular}{lrr}
    Ballot Characteristics & Count &  Percent \\
    \hline
    \multicolumn{3}{l}{\emph{Three most common ballots; no other appeared $>20{,}000$ times}}				\\[2pt]
Cuomo alone	&	$176{,}898$ 	&	 $16.5\%$ \\
$\mathrm{Mamdani}\prec \mathrm{Lander}\prec \mathrm{Adams}\prec \mathrm{Myrie}\prec \mathrm{Blake}$ 	&	$75{,}274$ 	&	 $7.0\%$ \\
$\mathrm{Mamdani}$ alone	&	$41{,}645$ 	&	 $3.9\%$ \\[6pt]
				
\multicolumn{3}{l}{\emph{The four most frequently ranked candidates}}				\\[2pt]
Ballots including Mamdani 	&	$639{,}693$ 	&	 $59.7\%$ \\
Ballots including Lander 	&	$631{,}838$ 	&	 $59.0\%$ \\
Ballots including Adams     &   $578{,}320$     &    $54.0\%$ \\
Ballots including Cuomo 	&	$489{,}682$ 	&	 $45.7\%$ \\[6pt]
				
\multicolumn{3}{l}{\emph{Prevalence of ranked voting by Mamdani and Cuomo voters}}				\\[2pt]
$>1$ candidate ranked including Mamdani but not Cuomo 	&	$485{,}071$ 	&	 $45.3\%$\\
$>1$ candidate ranked including Cuomo but not Mamdani 	&	$199{,}807$ 	&	 $18.7\%$\\[6pt]
				
\multicolumn{3}{l}{\emph{One common and three uncommon combinations of popular candidates}}				\\[2pt]
Ballots including Lander and Mamdani 	&	$502{,}479$ 	&	 $46.9\%$ \\
Ballots including Lander and Cuomo 	    &	$153{,}823$ 	&	 $14.4\%$ \\
Ballots including Mamdani but not Lander&	$137{,}214$	    & 	 $12.8\%$ \\
Ballots including Mamdani and Cuomo 	&	$112{,}977$ 	&	 $10.5\%$ \\
    \hline
    \end{tabular}
    \caption{Ballot groupings in the 2025 NYC Democratic mayoral primary election, out of $1{,}071{,}081$ ballots.}
    \label{tab:common_rankings}
\end{table}
Table \ref{tab:common_rankings} reveals additional patterns within voters' rankings of candidates. The most common specific ballot ranked only Cuomo. This choice was more than twice as common as the next most popular ranking, $\mathrm{Mamdani}\prec \mathrm{Lander}\prec \mathrm{Adams}\prec \mathrm{Myrie}\prec \mathrm{Blake}$, which in turn was more common than ballots cast for Mamdani alone. However, ranked choice voting allows for a vast number of permutations of choices, so examining only the most common ballots overlooks several important patterns.  

Figure \ref{fig:nyc_EDA} displays a stacked bar chart of resulting votes by rank level for each candidate. Despite appearing alone on more ballots than any other candidate, Cuomo was only the fourth-most commonly ranked candidate overall, well behind Mamdani, Lander, and Adams, who were the three top voter-getters in the progressive bloc. Notably, Lander was ranked on nearly as many ballots (59.0\%) as Mamdani (59.7\%). Two other patterns reveal the highly non-random differences in overall ranking behavior by voters with different preferences for their top ranking. First, the progressive bloc and Cuomo had distinct bases of support: just 10.5\% of ballots ranked both Mamdani and Cuomo and only 14.4\% ranked both Lander and Cuomo, but 46.9\% ranked both Mamdani and Lander---a large majority of voters ranking either Mamdani or Lander at all. Second, the progressive bloc had far more voters willing to vote for not just one candidate but a slate that excluded their main rival: 75.8\% of voters who selected Mamdani also ranked another candidate while omitting Cuomo, but just 40.8\% of Cuomo voters selected another candidate while omitting Mandani.
\begin{figure}[!ht]
    \centering
    \includegraphics[width=.8\textwidth]{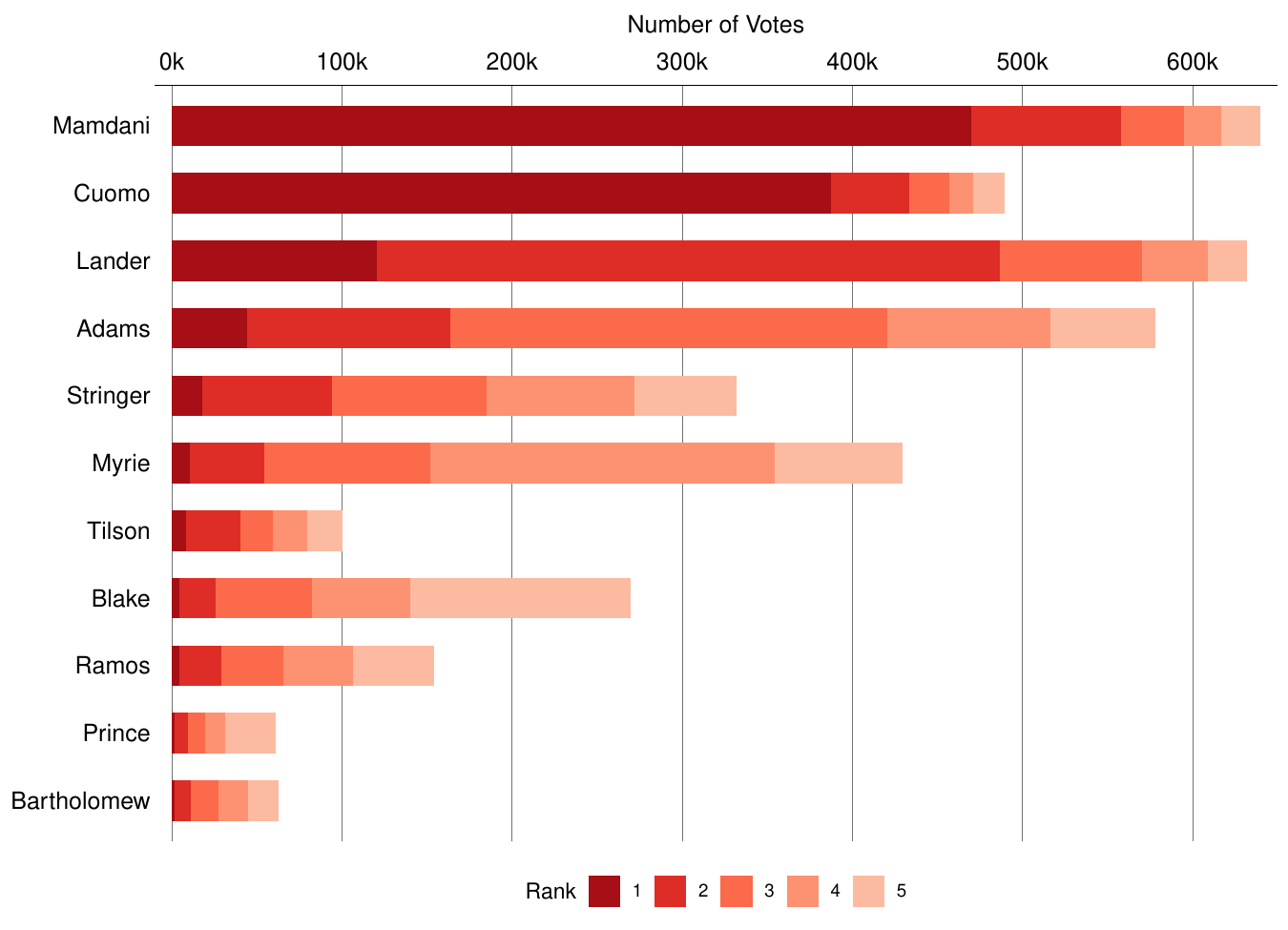}
    \caption{Votes in the 2025 NYC Democratic mayoral primary election by candidate and rank level. Candidates are ordered from top to bottom by the number of first-place votes received.}
    \label{fig:nyc_EDA}
\end{figure}

\subsection{Model estimation and comparison}

We now estimate and compare statistical models to the observed votes.\footnote{Code to replicate our analyses is available at 
REDACTED.
} Specifically, we consider the STV-R rationalizing model (equation \ref{model:generalized_stvr}) and latent class mixtures of either Plackett-Luce (equation \ref{eq:PL}) or Benter (equation \ref{eq:Benter}) distributions.\footnote{To our knowledge, the mixture of Benter distributions we fit is novel. Most relatedly, \cite{gormley2008exploring} fit a latent class mixture of Benter distributions with overall (i.e., non-class specific) dampening parameters.} Mixtures are estimated because we suspect voters belong to ideologically distinct factions, which may be detected under a mixture framework. For each distribution, we fit mixtures of between $K=1$ and $K=6$ latent classes (henceforth referred to as voting \textit{blocs}). We do not estimate the STV-R rationalizing model under a mixture framework as the rationalizing interpretation disappears beyond one bloc.

Figure \ref{fig:nyc_modelselectioncriteria} displays BIC and ICL criteria values (equation \ref{eq:bicicl}) for each fitted model. 
The STV-R rationalizing model performs worse than Plackett-Luce and Benter models under any number of latent classes. BIC consistently favors Benter models over Plackett-Luce. BIC favors Benter models with more latent classes, with flattening decreases beyond 2 latent classes.
Alternatively ICL is smallest in a Benter with 4 latent classes, although the 3-class Benter has a similarly small ICL.
On the basis of these criteria and in the interest of parsimony, we select the 3-class mixture of Benter distributions for further study. In what follows, we present results and examine the fit of the 3-class mixture of Benter distributions and the STV-R rationalizing model based on their estimated parameters. While the STV-R does not fit the data nearly as well as the 3-class Benter model, we examine it in detail for illustrative purposes since it has not been used in practice before.
\begin{figure}[!ht]
    \centering
    \includegraphics[width=.8\textwidth]{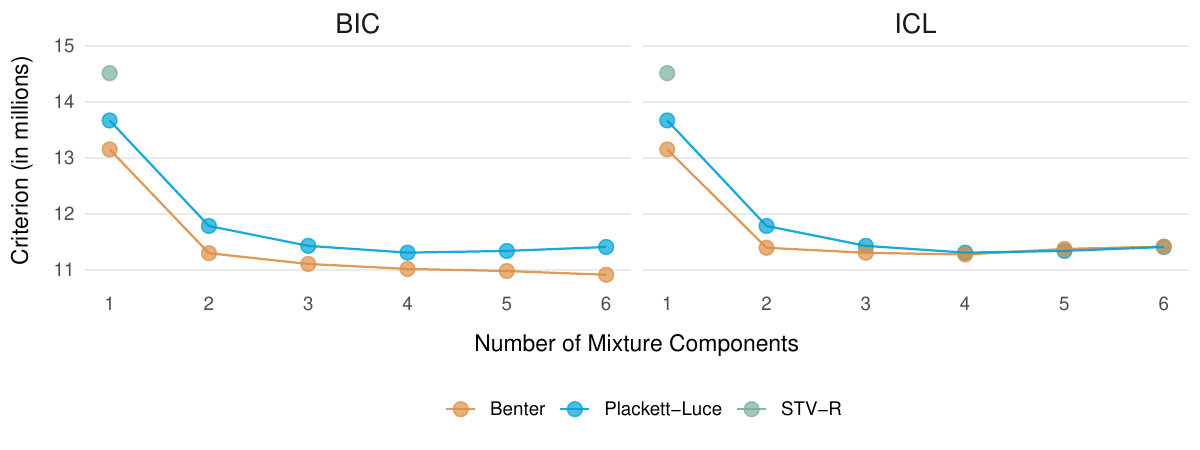}
    \caption{BIC and ICL model selection criteria for STV-R rationalizing model and latent class mixtures of Plackett-Luce and Benter models with up to 6 classes.}
    \label{fig:nyc_modelselectioncriteria}
\end{figure}

\subsection{Model fit and interpretation}

Figure \ref{fig:nyc_BenterK3} displays maximum likelihood estimates (MLE) of class-specific parameters in the 3-class mixture of Benter distributions. The left panel displays relative candidate worth parameters $\hat{p}_{jk}$, by candidate $j$ and class $k$. The right panel displays estimated dampening parameters, $\alpha_{rk}$, by rank level $r$ and class $k$.
\begin{figure}[tb]
    \centering
    \includegraphics[width=.8\textwidth]{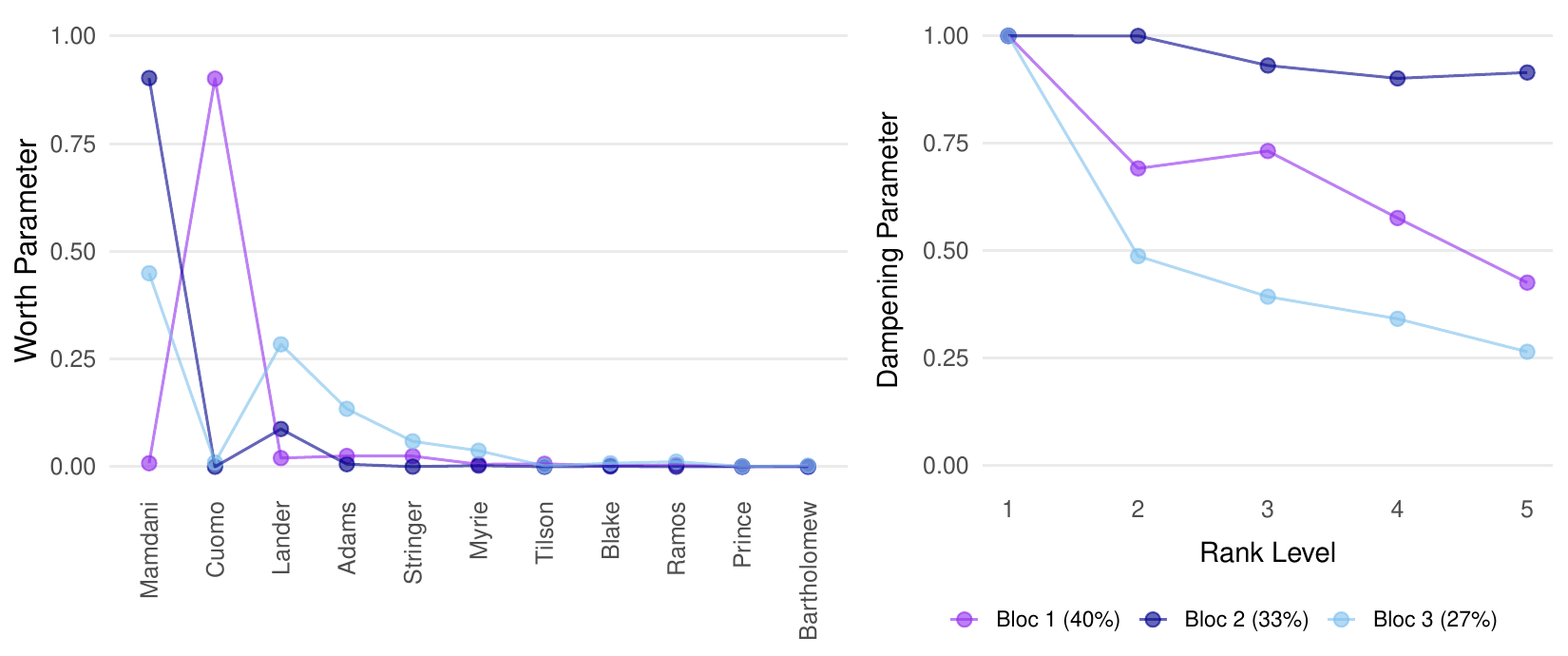}
    \caption{MLEs of candidate worth parameters (left) and rank-specific dampening parameters (right) by voter bloc (i.e., latent class) in a 3-class mixture of Benter distributions.}
    \label{fig:nyc_BenterK3}
\end{figure}
Table \ref{fig:nyc_mle} displays MLEs of model parameters, $\hat{k}$, by candidate for the STV-R rationalizing model. The order of $\hat{k}$ from largest to smallest is the STV-R ranking, $s$, by construction.
\begin{table}[!ht]
    \centering
    \begin{tabular}{l|c||l|c||l|c}
    Candidate & $\hat{k}$ & Candidate & $\hat{k}$ & Candidate & $\hat{k}$ \\
    \hline
  Mamdani   & $e^{-0.001}$    & Stringer  & $e^{-2.428}$  & Ramos       & $e^{-5.990}$   \\
  Cuomo     & $e^{-0.044}$    & Myrie     & $e^{-2.653}$  & Prince      & $e^{-8.326}$   \\
  Lander    & $e^{-1.002}$    & Tilson    & $e^{-3.100}$  & Bartholomew & $e^{-11.58}$   \\  
  Adams     & $e^{-1.822}$    & Blake     & $e^{-4.309}$  & &
    \end{tabular}
    \caption{MLEs of model parameters, $\hat k$, based on the STV-R rationalizing model. The order of candidates based on $\hat{k}$ is the STV-R ranking by construction.}
    \label{fig:nyc_mle}
\end{table}

We wish to assess if either model provides acceptable absolute fit to the observed votes based on goodness-of-fit criteria 1--3 proposed in section \ref{sec:STVR}. Recall that criterion 1 concerns the number of first place votes received by each candidate, criterion 2 the number of second place votes received by each candidate among voters who ranked Mamdani in first place, and criterion 3 the relative orderings of all candidate pairs. Figures \ref{fig:GOF1} and \ref{fig:GOF2} compare absolute and percent differences, respectively, in observed and expected counts for each criterion and model.
\begin{figure}[h!!]
    \centering
    \includegraphics[width=.8\textwidth]{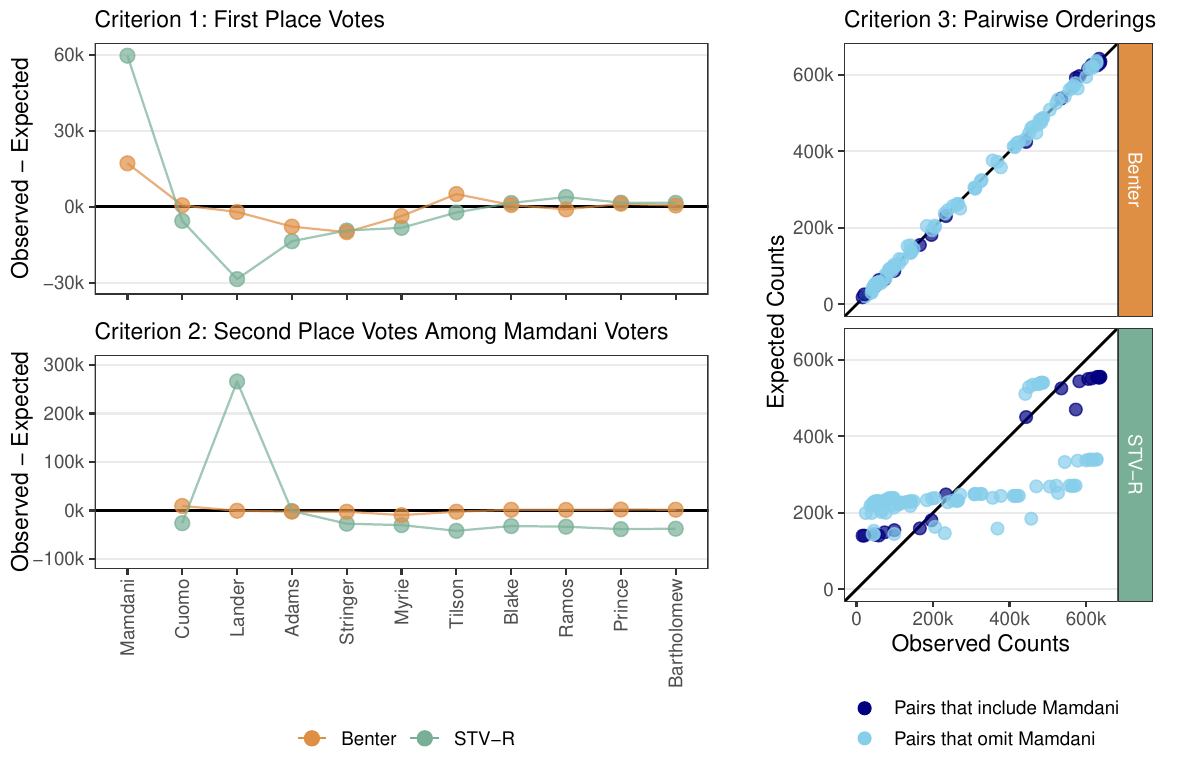}
    \caption{Comparison of STV-R rationalizing model and 3-class mixture of Benters based on difference between observed and expected counts across criteria 1--3. Solid black lines indicate perfect model fit.}
    \label{fig:GOF1}
\end{figure}
\begin{figure}[h!!]
    \centering
    \includegraphics[width=.8\textwidth]{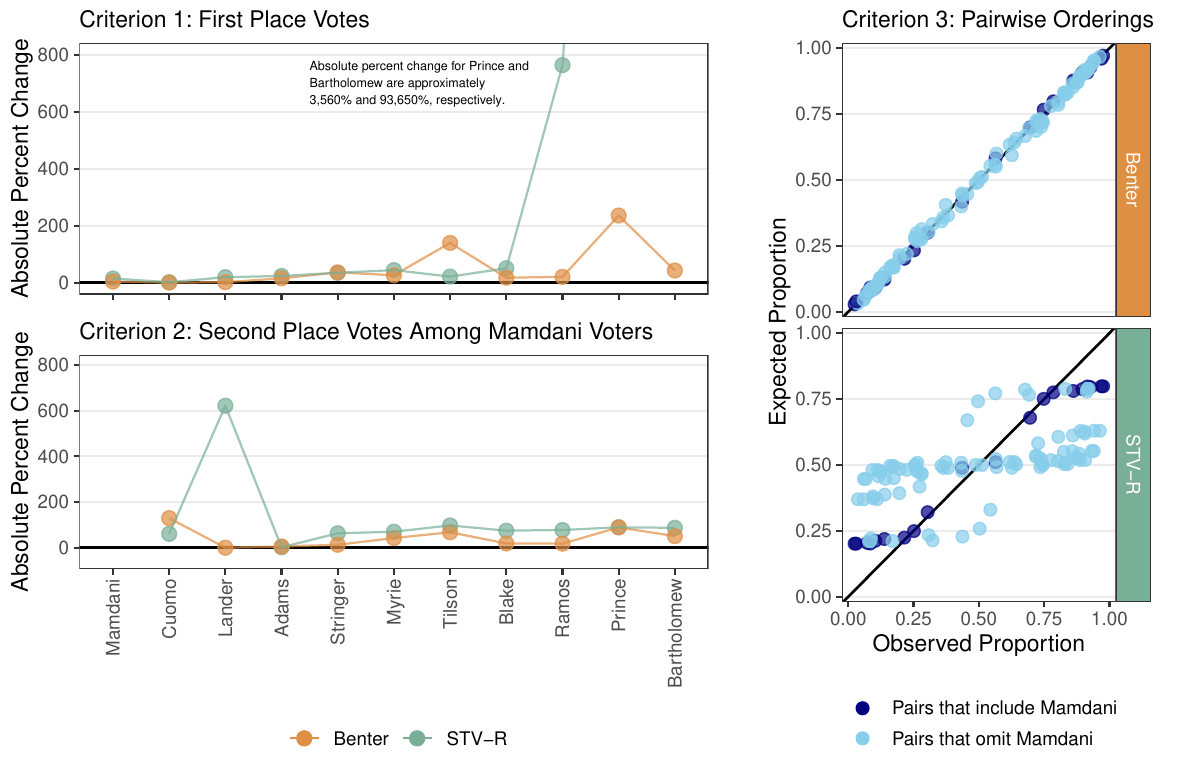}
    \caption{Comparison of STV-R rationalizing model and 3-class mixture of Benters based on percent change (criteria 1--2) or proportions (criterion 3) between observed and expected counts. Solid black lines indicate perfect model fit.}
    \label{fig:GOF2}
\end{figure}

We find that the 3-class Benter model provides acceptable fit to observed vote data on the basis of all three criteria. For criterion 1, observed and expected counts of first-place votes are similar for all candidates. The difference is largest in magnitude for Mamdani, for whom the percent error rate is only 3.8\%. For criteria 2 and 3, observed and expected counts are nearly identical for all candidates and candidate pairs.

In contrast, the STV-R rationalizing model poorly fits the observed votes on the basis of all three criteria. 
For criterion 1, observed and expected first-place vote counts differ by approximately $60$k and $29$k votes, respectively, for Mamdani and Lander. 
For criterion 2, Lander received approximately $266$k more votes in second place than expected among voters who ranked Mamdani first, indicating the model's inability to capture the large bloc of voters strategically supporting the progressive candidate alliance. Strikingly, the STV-R model fit based on this STV-R-specific criterion is worse for all but one candidate (Adams) than that of the Benter mixture model.
For criterion 3, observed pairwise orderings between candidates differ substantially from those expected under the model. Notably, pairs that include Mamdani are not well-captured by the model. These pairs are especially important both because Mamdani won the overall election and because he strategically allied with certain other candidates to reach that outcome.

Now that we have established the 3-class Benter provides acceptable fit, we interpret the estimated model parameters (refer to Figure \ref{fig:nyc_BenterK3}).
Recall that $\pi_k$ represents the proportion of voters in class $k$.
\textit{Worth} parameters $\theta_{jk}$ may be interpreted as the proportion of voters in bloc $k$ who prefer candidate $j$ in first place.
Then, \textit{dampening} parameters $\alpha_{rk}\in[0,1]$ may be interpreted as the inverse level of additional randomness introduced at rank level $r$ in bloc $k$, such that $\alpha_{rk}=1$ indicates no additional randomness and $\alpha_{rk}=0$ indicates each candidate has equal probability of selection.

Bloc 1 comprises approximately 40\% of the votes and highly favors Cuomo, while placing little worth on other candidates. This bloc represents support for the major centrist candidate. 
Blocs 2 and 3 comprise 33\% and 27\% of the vote, respectively, and collectively capture progressive voters. Both blocs favor Mamdani more than any other candidate, yet differ in the exclusiveness of that support: bloc 2 voters place only a small amount of worth on Lander, whereas bloc 3 voters more evenly split their support between Mamdani and Lander, and even reserve worth for Adrienne Adams. 
Dampening parameters are consistently high for bloc 2 and lower for blocs 1 and 3. Thus bloc 1 strongly favors Cuomo and exhibits weak preferences among the remaining candidates, bloc 2 strongly favors Mamdani and then Lander, and bloc 3 similarly favors Mamdani and Lander but with less strict preferences on the order of Mamdani and Lander.

\section{Discussion}\label{sec:discussion}

This paper investigates the intersection of social choice theory and statistical preference modeling, especially in non-epistemic settings applicable to voting in democratic elections, surfacing fundamental tensions that have previously gone unexamined. Statistical reasoning has been part of social choice theory since its inception \citep{condorcet1785essay}. More recently, the literature on \textit{rationalizability} has drawn precise connections between outcomes of social choice rules and statistical estimators of population preferences derived from a model \citep{young1995optimal,conitzer2009preference,pivato2013voting}. Examining this literature from a statistical vantage point reveals a deep structural mismatch: whereas social choice theory is principally concerned with aggregating individual preferences to produce winners, statistical preference modeling is more broadly oriented toward inference on individual preferences and the distributions thereof. We show that widely-used social choice rules for ranked choice elections either admit rationalizing models with little-to-no capacity to fit distributions of actual votes, or admit no rationalizing models at all. Conversely, we find empirically that well-fitting statistical models tend to be mixture models that have the capacity to capture systematic differences in preferences. However, mixture models offer no clear mechanism for making a social choice. Even a na\"ive weighting of estimated preferences among distinct voting blocs in a mixture model would be difficult to explain to voters, and, as such, is unlikely to be accepted as a social choice. Together, our findings suggest that social choice theory and statistical modeling may be considered as competing frameworks that offer distinct insights when applied to votes.

Figure \ref{fig:venn} summarizes the social choice rules and statistical models studied in this paper, where their intersection represents rationalizability. On the left, we see non-rationalizable social choice rules. Notably, these include instant runoff voting (IRV) and multi-winner single transferable vote (STV-K), which satisfy many desirable social choice axioms and special interest groups purport them to produce more representative outcomes \citep{benade2021ranked,richie2023case}. However, these non-rationalizable social choice rules are ultimately algorithms with no underlying generative models. Thus, inferring the multidimensional and complex ideologies of voters based on these rules is not possible.
\begin{figure}[h!]
    \centering
    \includegraphics[width=0.6\linewidth]{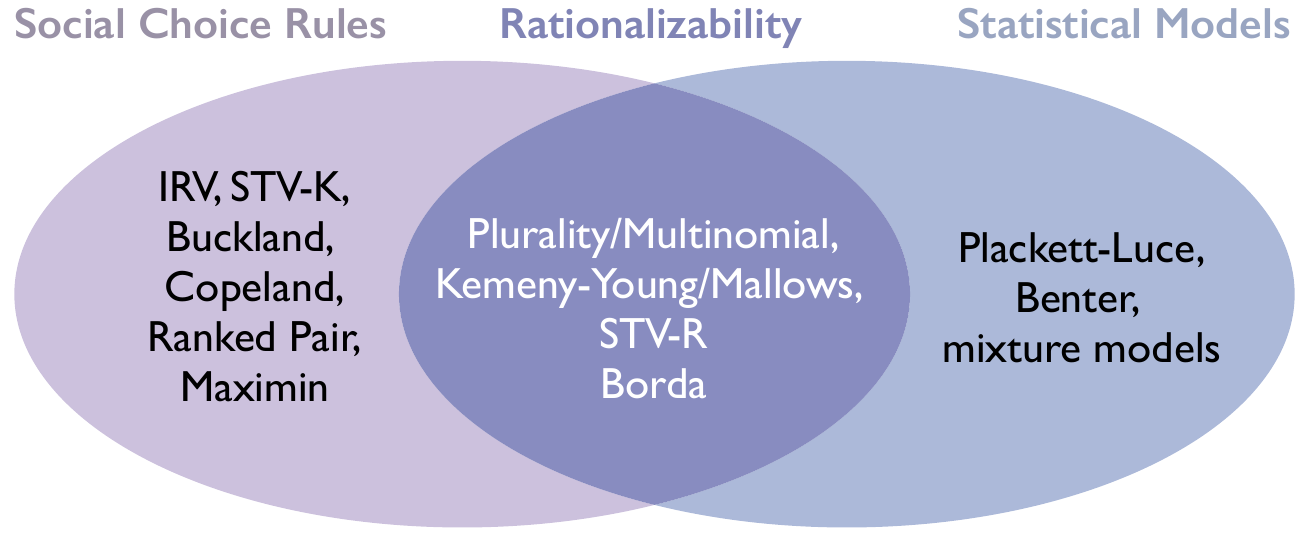}
    \caption{Venn diagram between social choice rules and statistical models, with their intersection representing rationalizability.}
    \label{fig:venn}
\end{figure}

On the right (in Figure \ref{fig:venn}), we see statistical models that have no corresponding social choice rules. These models include the Plackett-Luce and Benter distributions, perhaps estimated under a latent class mixture modeling framework. \textit{If they provide reasonable fit to data}, these models may aid inference on voter ideologies or coalitions. They can be useful for making interpretable or falsifiable statements about voter preferences, describing prominent voter blocs, or making probabilistic predictions of future elections. 
However, to be accepted as social choice rules, these models need to offer a clear and easily computable mechanisms for calculating winners. In the absence of such mechanisms, as in the case for Plackett-Luce and Benter distributions, statistical models  are not going to be adopted as societal decision-making tools.

The center in Figure \ref{fig:venn} displays social choice rules that are rationalized by statistical models. Rationalizable rules potentially offer a ``best of both worlds" approach to social choice: an axiomatically plausible way of determining a winner while simultaneously estimating, inferring, and predicting the preferences of a constituency. Unfortunately, this dual potential is not realized in any of the cases considered in this paper. Plurality rule and the Borda count, while widespread and simple, are highly susceptible to perverse social choice outcomes (most notably the spoiler effect). Kemeny-Young's complete ranking is difficult to explain to voters, and has been shown to be computationally expensive and unstable~\citep{ali2012experiments}, limiting its utility in real-world social choice settings.
Finally, this paper elucidated limitations of STV-R rationalizing models to fit realistic voting data. 

We demonstrate our claims empirically in a case study of the 2025 NYC Democratic primary election. We first highlight the influence of the ranked choice voting procedure in this election. In an election with many candidates, voters may be reasonably described as supporting the well-known but scandal-ridden centrist Andrew Cuomo or an array of progressive challengers (including the initially lesser-known Zohran Mamdani). In a traditional plurality vote election, the former group may have won due to vote splitting among the latter. Instead, IRV permitted a clear Mamdani win. That said, the tabulations alone miss an important feature of the voting population: although Cuomo came in second place, he is preferred second in very few voters' minds. Cuomo was ranked first on most ballots in which he appeared, and rarely ranked at all by voters supporting any progressive candidate (e.g., Mamdani, Lander, Adams, Myrie, Blake, or Stringer). Thus, we may consider the votes multi-modal, a characteristic lost in a social choice tabulation designed to make a single compromise ordering. As such, we found the unimodal STV-R rationalizing model to provide exceptionally poor fit to observed votes based on three goodness-of-fit criteria that capture important aspects of the votes, including first- and second-place votes and bivariate relationships between all pairs of candidates.

Continuing with a purely statistical analysis of votes in our case study, we found that a 3-class Benter model provided acceptable model fit based on the same criteria used to assess the rationalizing model. This is notable given that the model contains just 44 parameters to capture the highly diverse preferences of over 1 million voters who were permitted to express any of $64{,}471$ potential rankings over 11 candidates. The model-identified voter blocs are highly interpretable and align with political observers' expectations. Furthermore, estimated parameters and associated uncertainty permit various inferences on the relative prevalence of prominent voter ideologies and the relative strengths of candidates. But despite the posited model successfully providing a generative mechanism for observed votes, it is not immediately suitable for addressing the collective social choice problem which motivated the election as the inferred population heterogeneity cannot be straightforwardly reconciled.

This work highlights the need for careful statistical modeling of voting data. In realistic circumstances where no ``ground truth" exists and voters differ systematically in political preferences, voting tabulations are likely to miss important features of the distribution of individual opinions. Political scientists, journalists, and other individuals interested in understanding or communicating voter preferences should consider using not only mixture models like those used here, but also more complex statistical models. For example, mixed membership models may be used to model voters on an ideological spectrum \citep{airoldi2015handbook}, whereas rank-clustering models may be used to infer if groups of candidates are equally-preferred by subsets of voters \citep{pearce2024bayesian,piancastelli2025clustered}. Our work demonstrates a framework for statistical modeling of election data that could be applied
to aid nuanced inference on voter preferences.

\FloatBarrier

\bibliographystyle{abbrvnat}
\bibliography{main.bib}

\newpage
\appendix

\section*{Appendix A: Proofs of Theorems \ref{thm:irv} and \ref{thm:stv-k}}

\subsection*{Proof of Theorem \ref{thm:irv}}

Consider the voter profiles $V_1$ and $V_2$ shown in Figure~\ref{fig:STV_counter}.
When IRV is applied to $V_1$, candidates $a_4,\dots,a_J$ have no first-choice votes and are eliminated; $a_3$ is eliminated next. The three first-choice votes for $a_3$ are then transferred to $a_1$, resulting in $a_1$ being the winner with $7$ of the $13$ votes. For profile $V_2$, candidates $a_4\dots,a_J$ have no first-choice votes and are eliminated; $a_2$ is eliminated next. Again, this leaves $a_1$ as the winner with $7$ out of $13$ votes. In contrast when the two profiles are combined to obtain $V_1\cup V_2$, $a_1$ is eliminated before $a_2$ or $a_3$, and hence not the winner. Thus $f(V_1)=f(V_2)=a_1$ but $f(V_1\cup V_2)\neq a_1$, and so by Lemma~\ref{lem:reinforcment}, $f$ is not MLE-rationalizable. 
\begin{figure}[!ht]
    \centering
    $V_1=$\begin{tabular}{c|c}
         \# votes & $v$ \\
         \hline
         3 & $a_3\succ a_1\succ  a_2\succ a_4\succ\dots\succ a_J$\\
         4 & $a_1\succ  a_2\succ a_3\succ a_4\succ\dots\succ a_J$\\
         6 & $a_2\succ a_1\succ a_3\succ a_4\succ\dots\succ a_J$\\
    \end{tabular}\qquad
    $V_2=$\begin{tabular}{c|c}
         \# votes & $v$ \\
         \hline
         3 & $ a_2\succ  a_1\succ a_3\succ a_4\succ\dots\succ a_J$\\
         4 & $ a_1\succ a_3\succ  a_2\succ a_4\succ\dots\succ a_J$\\
         6 & $a_3\succ  a_1\succ  a_2\succ a_4\succ\dots\succ a_J$\\
    \end{tabular}\qquad
    $V_1\cup V_2$=\begin{tabular}{c|c}
         \# votes & $v$ \\
         \hline
         9 & $ a_2\succ  a_1\succ a_3\succ a_4\succ\dots\succ a_J$\\
         4 & $ a_1\succ a_3\succ  a_2\succ a_4\succ\dots\succ a_J$\\
         4 & $ a_1\succ  a_2\succ a_3\succ a_4\succ\dots\succ a_J$\\
         9 & $a_3\succ  a_1\succ  a_2\succ a_4\succ\dots\succ a_J$\\
    \end{tabular}
    \caption{Voter profiles to demonstrate that IRV is not MLE-rationalizable.}
    \label{fig:STV_counter}
\end{figure}

\subsection*{Proof of Theorem \ref{thm:stv-k}}

Consider the voter profiles $V_1$ and $V_2$ shown in Figure \ref{fig:STVK_example}. Let $e_1,\dots,e_{k}$ be the $k$ candidates who will be elected in each of $V_1$ and $V_2$. Let $a_1, a_2, \dots, a_{J-K}$ be the remaining candidates.
\begin{figure}[!b]
    \centering
     \begin{subfigure}[t]{0.3\textwidth}
    \centering
    \begin{tabular}{c|c}
        \# votes & $v$ \\
         \hline
         7 & $e_1 \succ \dots $\\
         7 & $e_2 \succ \dots $\\
         \vdots & \vdots \\
         7 & $e_{k-1} \succ \dots $\\
         4 & $e_{k}\succ a_2\succ \dots$\\
         3 & $a_1\succ e_k\succ \dots$\\
         6 & $a_2\succ e_k\succ \dots$\\
    \end{tabular}
    \caption{Profile $V_1$}
    \label{fig:STVK_a}
    \hfill
    \end{subfigure}
    \begin{subfigure}[t]{0.3\textwidth}
    \centering
    \begin{tabular}{c|c}
         \# votes & $v$ \\
         \hline
         7 & $e_1 \succ \dots $\\
         7 & $e_2 \succ \dots $\\
         \vdots & \vdots \\
         7 & $e_{k-1} \succ \dots $\\
         4 & $e_{k}\succ a_1\succ \dots$\\
         3 & $a_2\succ e_k\succ \dots$\\
         6 & $a_1\succ e_k\succ \dots$\\
    \end{tabular}
    \caption{Profile $V_2$}
    \label{fig:STVK_b}
    \end{subfigure}
    \begin{subfigure}[t]{0.3\textwidth}
    \centering
    \begin{tabular}{c|c}
         \# votes & $v$ \\
         \hline
         14 & $e_1 \succ \dots $\\
         14 & $e_2 \succ \dots $\\
         \vdots & \vdots \\
         14 & $e_{k-1} \succ \dots $\\
         4 & $e_{k}\succ a_1\succ \dots$\\
         4 & $e_{k}\succ a_2\succ \dots$\\
         9 & $a_2\succ e_k\succ \dots$\\
         9 & $a_1\succ e_k\succ \dots$\\
    \end{tabular}
    \caption{Combined Profile $V_1 \cup V_2$}
    \label{fig:STVK_c}
    \end{subfigure}
    \caption{Voter profiles to demonstrate that STV-K is not MLE-rationalizable. Ellipses are used to denote a preference ordering among all remaining candidates (specific candidate orderings donot impact results).}
    \label{fig:STVK_example}
\end{figure}
    
In each profile $V_1$ and $V_2$ there are $7K+6$ votes and hence, 7 votes is the threshold for electing a candidate in an election with $K$ winners. In both $V_1$ and $V_2$, candidates $e_1,\dots, e_{k-1}$ are elected on the first round; no votes are transferred. In profile $V_1$, $a_1$ has the fewest first-place votes at the end of round 1 and is removed, with votes transferred to the second-most preferred candidate, $e_k$. Then, $e_k$ is elected with 7 votes. Similarly, in profile $V_2$, $a_2$ has the fewest first-place votes at the end of round 1 and is removed, with votes transferred to the second-most preferred candidate, $e_k$. Then, $e_k$ is elected with 7 votes. In short, profiles $V_1$ and $V_2$ elect candidates $e_1,\dots, e_k$.

In the combined profile, $V_1 \cup V_2$, there are $14K+12$ votes, and hence $14$ votes is the threshold for electing a candidate in an election with $K$ winners. Candidates $e_1,\dots,e_{k-1}$ are elected on the first round; no votes are transferred. Then, $e_k$ has the fewest first place votes and is removed. Thus, the outcome of our social choice rule is different for the combined profile $V_1\cup V_2$, proving that STV-K is not MLE-rationalizable by Lemma \ref{lem:reinforcment}.

\section*{Appendix B: Further Details on the STV-R Rationalizing Model}

We first prove a sufficient separation condition for parameters $k=(k_1,\dots,k_J)$ in the STV-R rationalizing model to ensure that the STV-R ranking, $s$, equals the MLE ranking, $\hat s$. Then, we consider how the separation condition may be relaxed in practice while still ensuring that $\hat s=s$.

\begin{lemma}
    Consider the STV-R rationalizing model shown in equation \ref{model:generalized_stvr}. Denote by $s$ the STV-R ranking. Then if $k_j<(\prod_{\ell=1}^{j-1} k_\ell)^N$ for all $j=2,\dots,J$, then the model MLE $\hat s$ is $s$.
\end{lemma}
\begin{proof}
    Let $k$ be fixed such that $k_j<(\prod_{\ell=1}^{j-1} k_\ell)^N$ for all $j=2,\dots,J$. For contradiction, assume that $\hat s=s'$ such that $s'\neq s$. 

    For convenience, let $N_j(s)= \sum_{i=1}^N \delta_j(v_i,s)$ for all $j=1,\dots,J$. Note that by construction of $s$, $(N_J(s),N_{J-1}(s),\dots,N_1(s))$ is lexicographically smaller than $(N_J(s'),N_{J-1}(s'),\dots,N_1(s'))$. Thus, let $j^*$ denote the largest index in which $N_j(s)\neq N_j(s')$. Therefore, $N_{j^*}(s')\geq N_{j^*}(s)+1$. Then,
    \begin{align}
        \log p(V|s',k) &= \sum_{j=1}^J \log k_j\underbrace{\sum_{i=1}^N \delta_j(v_i,s')}_{=N_j(s')} - \underbrace{\sum_{i=1}^N \log C_{r_i}(k)}_{\equiv C}\\
        &= \sum_{j=1}^{j^*-1} N_j(s')\log k_j+ N_{j^*}(s')\log k_{j^*}+\sum_{j=j^*+1}^J N_j(s)\log k_j - C\\
        &< 0+ \log k_{j^*}(N_{j^*}(s)+1)+\sum_{j=j^*+1}^J \log k_jN_j(s) - C \label{step1}\\
        &= \log k_{j^*}+\sum_{j=j^*}^J \log k_jN_j(s) - C \\
        &\leq \sum_{j=1}^{j^*-1}N\log k_j+\sum_{j=j^*}^J \log k_jN_j(s) - C \label{step2}\\
        &\leq \sum_{j=1}^{j^*-1}N_j(s)\log k_j+\sum_{j=j^*}^J \log k_jN_j(s) - C \label{step3}\\
        &=\log p(V|s,k).
    \end{align}
Line \ref{step1} holds since $k_j\in(0,1)$ implies $\log k_j<0$. 
Line \ref{step2} holds since $k_j<(\prod_{\ell=1}^{j-1} k_\ell)^N$ implies $\log k_j < N\sum_{\ell=1}^{j-1}\log k_\ell$.
Line \ref{step3} holds since $N_j(s)\leq N$ for all $j$. We have thus shown $\log p(V|s',k)<\log p(V|s,k)$, which contradicts the assumption that $s'$ is the MLE.
\end{proof}

The above separation condition is sufficient, but not tight. In practice, we find a less strict separation constraint that permits a better model fit while still ensuring that $\hat s$ is the STV-R ranking $s$. To do so, we find some small $K$ such that $k_j<(\prod_{\ell=1}^{j-1} k_\ell)^K$ for all $j=2,\dots,J$ while maintaining $\hat s=s$. We find such a $K$ as follows:
\begin{enumerate}
    \item Set $K=\epsilon$, $\epsilon>0$.
    \item Estimate model parameters $k$ via maximum likelihood such that $s$ is the STV-R ranking and subject to constraint $k_j<(\prod_{\ell=1}^{j-1} k_\ell)^K$ for all $j=2,\dots,J$. Call the estimated $k$ parameters $\hat{k}_K$.
    \item Consider $\mathcal{S}=\{s' | 1\leq d_\tau(s',s)\leq 3\}$, where $d_\tau$ is the Kendall-$\tau$ distance. Calculate $\log p(V|\hat{k}_K,s')$ for all $s'\in\mathcal{S}$.
    \item If $\log p(V|\hat{s}_K,s)>\log p(V|\hat{s}_K,s')$ for all $s'\in\mathcal{S}$, stop the procedure and set $(\hat s,\hat k) = (s,\hat k_K)$. Otherwise, raise $K$ slightly and repeat steps 2--4.
\end{enumerate}
The above algorithm aims to find the less restrictive constraint on parameters $k$ such that the STV-R ranking stays the model MLE for parameter $s$. We note that the neighborhood of rankings $\mathcal{S}$ was kept relatively small for computational efficiency, although a larger neighborhood may be used if desired. In our Case Study, the separation constant was found to be $K=0.39$ by initializing $K$ at $0.01$ and incrementing by $0.01$.

\section*{Appendix C: Further Details on Goodness-of-Fit Criteria}

We proposed three goodness-of-fit criteria in section \ref{sec:STVR}. To assess a model based on these criteria, one must calculate the observed counts for each criteria under a given set of votes and the expected counts for each criteria under a fitted statistical model. The former is straightforward; one need only count the number of votes satisfying each criteria, e.g., the number of votes in which a given candidate is ranked in first place. 

The latter requires more careful explanation. Note that in our modeling exercises, we assume that the number of votes of each length are fixed and known. Thus, expected counts under each goodness-of-fit statistics are first calculated separately for votes of each length, and then an overall expected count for each criterion is computed as a weighted sum of the counts for votes of each length. Weights are determined by the relative size of each vote length. That is, the expected count, $E$, for some statistic under a fitted model is calculated:
$$E = \sum_r \frac{N_r}{N} E_r$$
where $r$ indexes the observed vote lengths, $N_r$ is the number of observed votes of length $r$, $N$ is the total number of votes, and $E_r$ is the expected counts for the statistic among only votes of length $r$ under the fitted model. 

We calculate $E_r$ by explicitly calculating the mass of each possible ranking of length $r$ under the fitted model, summing the mass of rankings satisfying the goodness-of-fit criterion, and multiplying the summed mass by $N_r$. In cases where the number of possible votes to prohibitively large as to compute $E_r$ exactly, Monte Carlo approximation may be used by sampling rankings under the fitted model.

\end{document}